\documentclass[letterpaper, 10 pt, conference]{ieeeconf}  % Comment this line out
\IEEEoverridecommandlockouts                              % This command is only
\usepackage{amsmath,amssymb,amsfonts}
\usepackage{algorithmic}
\usepackage{graphicx} 
\graphicspath{{./figs}}
\usepackage{textcomp}
\usepackage{xcolor}
\usepackage[hidelinks]{hyperref}
\usepackage{enumerate}
\usepackage{multirow}
\usepackage{gensymb}
\usepackage{cite}
\usepackage{booktabs}
\usepackage{colonequals}
\usepackage[font=footnotesize]{caption}
\usepackage{subcaption}
\usepackage{siunitx} % display and formatting of units
\usepackage[capitalize]{cleveref}
\crefname{figure}{Fig.}{Figs.}
\crefname{assumption}{Assumption}{Assumptions}
\crefname{definition}{Definition}{Definitions}

\usepackage{cuted}
\usepackage{flushend}
\usepackage{accents}
\usepackage{tablefootnote}
\newcommand{\bmat}[1]{\begin{bmatrix}#1\end{bmatrix}}

\newcommand{\unbar}[1]{\underaccent{\bar}{#1}}
\newcommand{\defeq}{\colonequals}
\def\BibTeX{{\rm B\kern-.05em{\sc i\kern-.025em b}\kern-.08em
    T\kern-.1667em\lower.7ex\hbox{E}\kern-.125emX}}

\renewcommand{\top}{\mathsf{T}} % prettier symbol for transpose

\newcommand {\R}{\mathbb{R}}

\DeclareMathOperator*{\argmin}{argmin}
\newtheorem{theorem}{\textbf{Theorem}}

\newtheorem{definition}{\textbf{Definition}}

\newtheorem{proposition}{\textbf{Proposition}}

\newcommand{\dd}{\mathrm{d}}
\renewcommand{\epsilon}{\varepsilon}

\title{\LARGE \bf
Constrained Feedback Control of Nonlinear Systems via Approximate HJB and Control Barrier Functions
}
\author{Milad Alipour Shahraki and Laurent Lessard% <-this % stops a space
	\thanks{M. Alipour Shahraki is with the Department of Electrical and Computer Engineering,
		Northeastern University, Boston, MA 02115, USA
		{\tt\small alipourshahraki.m@northeastern.edu}}%
	\thanks{L. Lessard is with the Department of Mechanical and Industrial Engineering,
		Northeastern University, Boston, MA 02115, USA
		{\tt\small l.lessard@northeastern.edu}}%
}
\begin{document}
\maketitle
\thispagestyle{empty}
\pagestyle{empty}
%%%%%%%%%%%%%%%%%%%%%%%%%%%%%%%%%%%%%%%%%%%%%%%%%%%%%%%%%%%%%%%%%%%%%%%%%%%%%%%%
\begin{abstract}
This paper presents a two-stage framework for constrained feedback control of input-affine nonlinear systems. Offline, an approximate value function for the unconstrained problem is computed, for example using Hamilton--Jacobi--Bellman (HJB)-based policy iteration. Online, the proposed quadratic program (QP) minimizes the pre-Hamiltonian evaluated using the approximate value-function gradient subject to safety constraints enforced by control barrier functions (CBFs). This architecture decouples performance optimization from constraint enforcement, allowing constraints to be modified without recomputing the value function. As in CBF-QP architectures based on control Lyapunov functions (CLFs), safety is enforced as a hard constraint; however, the performance objective targets approximate optimality rather than a prescribed Lyapunov decay. Numerical results on a linear 2-state hovercraft and a nonlinear 9-state spacecraft attitude-control problem show agreement with the constrained open-loop optimal control problem (OCP) benchmark in the linear case, and performance close to the OCP benchmark, improving on CLF-based controllers, in the nonlinear case.
\end{abstract}

\vspace{3mm}
%%%%%%%%%%%%%%%%%%%%%%%%%%%%%%%%%%%%%%%%%%%%%%%%%%%%%%%%%%%%%%%%%%%%%%%%%%%%%%%%
\section{Introduction}\label{sec:intro}

The optimal feedback control of general nonlinear systems remains a central challenge in control theory. The Hamilton--Jacobi--Bellman (HJB) equation provides a theoretical characterization of the optimal solution, but its nonlinear partial differential equation (PDE) structure generally precludes analytical solutions, while the curse of dimensionality makes numerical approximation computationally prohibitive for high-dimensional systems~\cite{lewis2012optimal}. Practical systems must also satisfy safety-critical constraints, such as actuator limits, state bounds, and geometric exclusion zones, further complicating optimal controller design. Model predictive control (MPC) can address both performance and constraints, but the need to solve a nonlinear optimization problem online can limit real-time applicability for fast or high-dimensional systems.

One alternative is to address optimality and safety separately. On the optimization side, approximate dynamic programming methods seek tractable approximations to the HJB equation. Successive Galerkin Approximation (SGA)~\cite{beard1997galerkin} projects the generalized HJB equation onto a finite-dimensional basis, while sum-of-squares (SOS) methods~\cite{summers2013approximate,yang2023approximate} replace polynomial nonnegativity conditions with tractable semidefinite relaxations. These approaches are often combined with policy iteration~\cite{jiang2015global,wang2015approximate} to successively improve the value-function approximation and associated feedback policy. However, they typically address unconstrained control problems; incorporating state and input constraints directly into the offline value-function computation can substantially increase computational complexity.

On the constraint enforcement side, control barrier functions (CBFs)~\cite{ames2019control} and their high-order extensions (HOCBFs)~\cite{xiao2021high} convert suitable state constraints into affine inequalities in the control input, enabling real-time enforcement through quadratic programs (QPs). CBFs are commonly combined with control Lyapunov functions (CLFs) in CLF-CBF-QP architectures~\cite{ames2016control} to jointly address stability and safety. However, as observed in our prior work on CLF-CBF-QPs~\cite{shahraki2025spacecraft}, CLF-based formulations introduce several design considerations. First, the CLF condition enforces stability rather than optimal performance and requires a suitable Lyapunov function, often hand-designed via input-output linearization~\cite{shahraki2025spacecraft}. Second, fixed CLF decay rates can induce input chattering at low sampling frequencies~\cite{shahraki2025spacecraft,dihel2024flexible}. Third, selecting the comparison function governing the prescribed decay introduces a feasibility-convergence trade-off that is further complicated by input constraints~\cite{fan2025concave}. These considerations motivate replacing the prescribed CLF decay objective with one derived from optimal control.

Another line of work couples safety and performance by incorporating CBF conditions directly into the value-function computation~\cite{cohen2020approximate,yazdani2020safety,almubarak2021hjb}. However, because the constraints are embedded in the offline optimization, changing the safe set generally requires recomputing the value function, reducing the ability to adapt to real-time changes in constraints.

Our approach uses a two-stage architecture that separates performance optimization from safety enforcement. Offline, an approximate value function for the unconstrained problem is computed. Online, its gradient defines the objective of a CBF-QP that enforces safety and input constraints. The controller therefore depends only on the value-function gradient, not on how the approximation is obtained. Our main contributions are as follows.
\begin{enumerate}[1.]
\item We propose the \emph{H-CBF-QP}, a convex QP that minimizes the approximate pre-Hamiltonian induced by a precomputed value function while enforcing CBF and input constraints. Unlike CLF-CBF-QP architectures~\cite{ames2016control,shahraki2025spacecraft}, which impose a prescribed Lyapunov decay condition that may be relaxed for safety, the H-CBF-QP favors the unconstrained value-function policy whenever feasible. This removes the need for a prescribed CLF decay rate and slack variable, as well as the input-output-linearization-based CLF construction.

\item We characterize the safety, approximate optimality, and local stability properties of the resulting closed loop. The CBF constraints provide forward-invariance guarantees, while the controller recovers the unconstrained value-function policy whenever the constraints are inactive. Under a regional HJB inequality, asymptotic stability is guaranteed on forward-invariant sublevel sets of the approximate value function contained in the inactive-constraint region.

\item We validate our proposed controller on a linear 2-state hovercraft and a nonlinear 9-state spacecraft attitude-control problem with reaction-wheel momentum and pointing constraints. The results agree with the constrained OCP benchmark in the linear example, achieve near-OCP performance under momentum constraints, and outperform CLF-CBF-QP benchmarks under both spacecraft constraints while maintaining CBF-based safety and real-time QP implementation.
\end{enumerate}

The remainder of the paper is organized as follows. \Cref{sec:problem} formulates the problem and reviews HJB and CBF concepts. \Cref{sec:H_cbf} develops the online H-CBF-QP and establishes its safety, approximate optimality, and stability properties. \Cref{sec:value_approx_section} reviews offline value-function approximation methods based on policy iteration. \Cref{sec:examples} validates the framework on linear and nonlinear examples, and \Cref{sec:conclusion} concludes with future research directions.

%%%%%%%%%%%%%%%%%%%%%%%%%%%%%%%%%%%%%%%%%%%%%%%%%%%%%%%%%%%%%%%%%%%%%%%%%%%%%%%%
\section{Problem Formulation}\label{sec:problem}

\subsection{System Dynamics and Objectives}
We consider a nonlinear \emph{input-affine} control system
\begin{equation}\label{eq:dynamics}
\dot{x} = f(x) + g(x)u, \qquad x(0)=x_0,
\end{equation}
where $x(t)\in\R^n$ is the state, $u(t)\in\R^m$ is the control input, and $f:\R^n\to\R^n$ and $g:\R^n\to\R^{n\times m}$ are locally Lipschitz continuous, with $f(0)=0$. We aim to minimize the infinite-horizon cost
\begin{equation}\label{eq:cost}
J(x_0, u(\cdot)) = \int_0^{\infty} \bigl( q(x) + u^\top R\, u \bigr)\, \dd t
\end{equation}
subject to the dynamics \eqref{eq:dynamics}, state constraint $x(t) \in \mathcal{X}\subseteq\R^n$, and input constraint $u(t)\in\mathcal{U}\subseteq\R^m$, where
$\mathcal{U}$ is a nonempty convex polyhedron that includes the origin.
Here, $q:\R^n\to\R_{\geq 0}$ satisfies $q(0)=0$ and $R\succ0$.

\subsection{Hamilton--Jacobi--Bellman (HJB) Equation}\label{sec:hjb}
We first consider the \emph{unconstrained} problem.
Let $\mathcal{P}$ denote the set of positive definite, proper, $C^1$ functions. Define the optimal value function
$
V^*(x_0)\defeq\inf_{u(\cdot)}J(x_0,u(\cdot)).
$
Assuming $V^*\in\mathcal{P}$, it satisfies the HJB equation
\begin{equation}\label{eq:hamiltonian}
H^*(x,\nabla V^*(x)) \defeq \min_{u\in\R^m}\; H(x,\nabla V^*(x), u) = 0,
\end{equation}
where $H^*$ denotes the Hamiltonian and the pre-Hamiltonian $H$ is defined as:
\[
H(x,\nabla V(x),u) \defeq \nabla V(x)^\top \bigl(f(x) + g(x) u\bigr) + q(x) + u^\top R u.
\]
The optimal (unconstrained) feedback law is given by
\begin{equation}\label{eq:optimal_control}
u^*(x)\defeq-\tfrac{1}{2}R^{-1}g(x)^\top\nabla V^*(x).
\end{equation}

\subsection{Safety via Control Barrier Functions (CBFs)}\label{sec:cbf_background}
We briefly review the CBF and HOCBF tools used for online constraint enforcement. For a scalar function $B$, define the Lie derivatives $L_f B\defeq\nabla B^\top f$ and $L_g B\defeq\nabla B^\top g$, so that $\dot B=L_f B+L_g B\,u$ along \eqref{eq:dynamics}.

\begin{definition}[Relative Degree~\cite{xiao2021high}]\label{def:reldeg}
A sufficiently differentiable function $B:\R^n\to\R$ has relative degree $r$ with respect to \eqref{eq:dynamics} if $L_gL_f^kB=0$ for $k=0,\ldots,r-2$ and $L_gL_f^{r-1}B\neq0$.
\end{definition}

\begin{definition}[HOCBF~\cite{xiao2021high}]\label{def:hocbf}
Let $B:\R^n\to\R$ have relative degree $r$. Using linear class-$\mathcal{K}$ functions with gains $\alpha_k>0$, define
\begin{equation}\label{eq:hocbf_seq}
\psi_0\defeq B,\quad
\psi_k\defeq\dot{\psi}_{k-1}+\alpha_k\psi_{k-1},
\quad k=1,\ldots,r,
\end{equation}
and $\mathcal{C}_k\defeq\{x\mid\psi_k(x)\geq0\}$. Then $B$ is a HOCBF if $\sup_{u\in\mathcal{U}}\psi_r(x,u)\geq0$ for all $x\in\bigcap_{k=0}^{r-1}\mathcal{C}_k$.
\end{definition}

\begin{theorem}[HOCBF Forward Invariance~\cite{xiao2021high}]\label{thm:hocbf_safety}
If $B$ is a HOCBF and $u(x)\in\mathcal{U}$ satisfies $\psi_r(x,u)\geq0$ on $\bigcap_{k=0}^{r-1}\mathcal{C}_k$, then this set is forward invariant.
\end{theorem}

For a relative-degree-$r$ function $B$, the HOCBF condition $\psi_r(x,u)\geq0$ is affine in $u$ and can therefore be imposed directly in a QP. In particular, for $r=1$ it reduces to $L_fB+L_gB\,u\geq-\alpha_1B$, while for $r=2$ it becomes $L_f^2B+L_gL_fB\,u+(\alpha_1+\alpha_2)L_fB+\alpha_1\alpha_2B\geq0$.

%%%%%%%%%%%%%%%%%%%%%%%%%%%%%%%%%%%%%%%%%%%%%%%%%%%%%%%%%%%%%%%%%%%%%%%%%%%%%%%%
\section{Controller Design}\label{sec:H_cbf}
We now present the H-CBF-QP, which combines an approximate value function for the \emph{unconstrained} problem (computed offline) with the CBF constraints of \Cref{sec:cbf_background}. Any approximation method may be used; specific examples are given in \Cref{sec:value_approx_section}.

Given the precomputed $\nabla \hat{V}$ and the current state $x$, we solve an online QP that minimizes the pre-Hamiltonian $H(x,\nabla \hat V(x),u)$ over $u$, subject to the state constraints (via the CBF/HOCBF) and input constraints:
\begin{equation}\label{eq:H_cbf_qp}
\begin{aligned}
\bar u(x) \;\defeq\; \argmin_{u} \;\; & H(x,\nabla \hat V(x),u)\hspace{-4mm}\\
\text{s.t.} \;\; & \psi_r(x, u) \geq 0, & \text{(CBF/HOCBF)} \\
& u\in\mathcal{U}. & \text{(Input constraint)}
\end{aligned}
\end{equation}
This is a pointwise optimization.
Since $\nabla \hat V(x)^\top f(x)$ and $q(x)$ do not depend on $u$, we may equivalently minimize 
\begin{equation}\label{eq:H_obj}
J_{\text{H}}(u) \defeq \nabla \hat{V}(x)^\top g(x)\, u + u^\top R\, u.
\end{equation}
To mirror the notation defined in \eqref{eq:optimal_control}, we let $\hat u$ denote the unconstrained minimizer of the pre-Hamiltonian $H(x,\nabla \hat V(x),u)$:
\begin{equation}\label{eq:uhat}
\hat u(x) \defeq -\tfrac{1}{2} R^{-1} g(x)^\top \nabla \hat{V}(x).
\end{equation}

\subsection{Theoretical Properties}
Under suitable assumptions, we can prove safety, approximate optimality, and asymptotic stability properties of our proposed controller. Our first result follows directly from strict convexity of \eqref{eq:H_obj}.
\begin{proposition}[Unconstrained Policy Recovery]\label{prop:recovery}
If $\hat u(x)$ is feasible for \eqref{eq:H_cbf_qp}, then the QP solution satisfies
$\bar u(x)=\hat u(x)$.
\end{proposition}
Thus, wherever the constraints are inactive, the online controller matches the unconstrained value-function policy. In this case, the deviation from the optimal policy depends on the deviation of the approximate value function gradient from the optimal one.
\begin{proposition}[Safety and Approximate Optimality]\label{thm:decoupling}
Suppose $B$ is a relative degree-$r$ HOCBF and $\mathcal{C}_k$ is given by \cref{def:hocbf}. If 
\eqref{eq:H_cbf_qp} is feasible for  $x\in \bigcap_{k=0}^{r-1}\mathcal{C}_k$, then:
(i)~$\bigcap_{k=0}^{r-1}\mathcal{C}_k$ is forward invariant under the resulting policy; and
(ii)~on any compact set $\Omega\subset\R^n$ where $\hat u$ is feasible, there exists a constant $C>0$ such that
\begin{equation*}
\|\bar u-u^*\|_\Omega
\leq C\,\bigl\|\nabla\hat V-\nabla V^*\bigr\|_\Omega,
\end{equation*}
where $\|\phi\|_\Omega\defeq\sup_{x\in\Omega}\|\phi(x)\|$. Thus,
$\|\bar u-u^*\|_\Omega\to0$ as
$\|\nabla\hat V-\nabla V^*\|_\Omega\to0$.
\end{proposition}
\begin{proof}
Part~(i) follows directly from \cref{thm:hocbf_safety}, since feasibility of \eqref{eq:H_cbf_qp} ensures satisfaction of the HOCBF constraint.
For part~(ii), \cref{prop:recovery} gives $\bar u=\hat u$ on $\Omega$. Using
\eqref{eq:optimal_control} and \eqref{eq:uhat},
$
\bar u-u^*
=
-\tfrac{1}{2}R^{-1}g(x)^\top
(\nabla\hat V-\nabla V^*).
$
Taking the supremum norm over $\Omega$ gives the stated bound with
$C=\tfrac{1}{2}\|R^{-1}g^\top\|_\Omega$.
Since $g$ is continuous and $\Omega$ is compact, $C$ is finite, proving the convergence claim.
\end{proof}

Finally, if the approximate value function satisfies a local HJB inequality, the unconstrained policy also provides a local Lyapunov decrease. This yields an asymptotic stability guarantee for sublevel sets contained in the region where the constraints are inactive.

\begin{theorem}[Local Asymptotic Stability]\label{cor:stability_interior}
Suppose $\hat V$ is $C^1$ and positive definite on a compact set
$\Omega$ containing the origin, and satisfies
$H(x,\nabla\hat V(x),\hat u(x))\leq0$ on $\Omega$. Let $q$ be positive
definite. Define $\mathcal{I}\subseteq\R^n$ as the set of states where
$\hat u(x)$ is strictly feasible for \eqref{eq:H_cbf_qp}, and assume
$\mathcal{I}\cap\Omega$ contains a neighborhood of the origin. Then, for any
$c>0$ such that
$
\mathcal{S}\defeq\{x\mid\hat V(x)\leq c\}
\subseteq\mathcal{I}\cap\Omega,
$
the set $\mathcal{S}$ is forward invariant and contained in the region of
attraction of the origin.
\end{theorem}

\begin{proof}
On $\mathcal{I}$, \cref{prop:recovery} gives $\bar u=\hat u$. Therefore,
$\dot{\hat V} = H(x,\nabla\hat V,\hat u)-q(x)-\hat u^\top R\hat u
\leq -q(x)$ for $x \in \mathcal{I}\cap\Omega$, since
$H(x,\nabla\hat V,\hat u)\leq0$ and $R\succ0$.
Since $\mathcal{S}\subseteq\mathcal{I}\cap\Omega$ is a sublevel set of $\hat V$ and
$\dot{\hat V}\leq0$ throughout $\mathcal{S}$, the set $\mathcal{S}$ is forward invariant.
Because $\hat V$ and $q$ are positive definite,
$\dot{\hat V}<0$ for all $x\in\mathcal{S}\setminus\{0\}$. Hence, the origin is asymptotically stable and $\mathcal{S}$ is contained in its region of attraction.
\end{proof}

%%%%%%%%%%%%%%%%%%%%%%%%%%%%%%%%%%%%%%%%%%%%%%%%%%%%%%%%%%%%%%%%%%%%%%%%%%%%%%%%
\section{Value Function Approximation}\label{sec:value_approx_section}
The H-CBF-QP requires only the gradient $\nabla\hat V$ of a precomputed value-function approximation.
The safety guarantee is independent of the approximation quality.
The stability result of \cref{cor:stability_interior} additionally requires $\hat V$ to be $C^1$ and positive definite on a compact approximation domain $\Omega$ containing the origin, with $H(x,\nabla\hat V(x),\hat u(x))\leq0$ on $\Omega$.

We review two popular approaches: Successive Galerkin Approximation (SGA) \cite{beard1997galerkin}, and sum-of-squares (SOS) \cite{jiang2015global}. Both approaches parameterize $\hat V(x)=\sum_i c_i m_i(x)$ using polynomial basis functions $m_i$ and solve for the coefficients via \emph{policy iteration}: alternating between policy evaluation and policy improvement steps.
Both use the same policy improvement step:
$\hat u^{(k+1)}=-\tfrac{1}{2}R^{-1}g^\top\nabla \hat V^{(k)}$.
For policy evaluation, both methods work with the negative pre-Hamiltonian
\begin{equation*}
    L^{(k)}(x) \defeq -H(x,\nabla \hat V^{(k)}(x),\hat u^{(k)}(x)),
\end{equation*}
over a compact domain $\Omega$.

\subsection{Successive Galerkin Approximation}\label{sec:sga}
SGA performs policy evaluation by seeking $L^{(k)} \approx 0$. This is accomplished by solving a system of orthogonality equations, which are linear in the coefficients $c_i$: $\langle L^{(k)},m_i\rangle_\Omega = 0$, where the inner product is an integral over $x\in\Omega$. SGA is scalable to larger systems because each step only involves linear solves and the integrals may be pre-computed for the monomial bases. Under suitable completeness conditions, increasing the number of basis functions yields convergence to $V^*$. For a finite basis, intermediate iterates do not come with stability guarantees, unlike the SOS formulation below.

\subsection{Sum-of-squares}\label{sec:sos}
SOS performs policy evaluation by minimizing $\int_\Omega \hat V^{(k)}(x)\,\dd x$ subject to $L^{(k)} \geq 0$ and $\hat V^{(k)}$ positive definite on $\Omega$. For a semialgebraic domain $\Omega$, the regional (restricted to $\Omega$) inequalities are certified using SOS multipliers. These conditions yield a semidefinite program in the coefficients $c_i$. Since $\hat u^{(k+1)}$ minimizes the pre-Hamiltonian,
\begin{equation*}
H(x,\nabla \hat V^{(k)},\hat u^{(k+1)}) \leq  H(x,\nabla \hat V^{(k)},\hat u^{(k)}) = -L^{(k)}(x)\leq0
\end{equation*}
on $\Omega$. Thus, each feasible iterate satisfies the regional HJB inequality and is positive definite on $\Omega$, providing the Lyapunov conditions used in \cref{cor:stability_interior} on sublevel sets contained in $\Omega$ and guaranteeing local stability.

%%%%%%%%%%%%%%%%%%%%%%%%%%%%%%%%%%%%%%%%%%%%%%%%%%%%%%%%%%%%%%%%%%%%%%%%%%%%%%%%
\section{Application Examples}\label{sec:examples}
We validate the H-CBF-QP on two examples: a linear system with a known constrained-optimal benchmark (Example~1) and a nonlinear spacecraft system with relative-degree-one and relative-degree-two constraints (Example~2). Together, they demonstrate (i) the ability to modify constraints without recomputing the offline value function, (ii) cases in which the H-CBF-QP closely reproduces the constrained-optimal trajectory, and (iii) the performance gap that can arise when the constrained optimum requires anticipation not encoded in the unconstrained value function. In both examples, the running state cost is a positive definite quadratic, $q(\xi)\defeq\xi^\top Q\xi$, where $\xi$ denotes the state used in the corresponding offline value-function problem.

The closed-loop simulations are implemented in Julia on an Apple M2 Pro with 32\,GB of RAM and use a sampling rate of \qty{10}{\hertz}. The constrained open-loop optimal control problem (OCP) is solved by direct transcription using JuMP/Ipopt, while all QP-based controllers use JuMP/HiGHS. Constrained LQR and MPC provide additional benchmarks for Example~1. For Example~2, the Optimal-Decay (OD) and Rapidly-Exponentially-Stabilizing (RES) CLF-CBF-QP controllers from~\cite{shahraki2025spacecraft} are tuned for best performance. All methods use the same shared model and cost parameters.

%%%%%%%%%%%%%%%%%%%%%%%%%%%%%%%%%%%%%%%%%%%%%%%%%%%%%%%%%%%%%%%%%%%%%%%%%%%%%%%%
\subsection{Example 1 (Linear): 1D Hovercraft}
As a linear benchmark, we apply the H-CBF-QP to a unit-mass 1D hovercraft with state
$x=[p,\,v]^\top\in\R^2$ (position and velocity) and input $u\in\R$ (thrust force). The system is a double integrator of the form \eqref{eq:dynamics}, with
\begin{equation*}
f(x)\defeq\bmat{v\\0},\qquad
g(x)\defeq\bmat{0\\1}.
\end{equation*}
For this example, we used SGA as a proof of concept, but since this is a standard linear-quadratic problem, the optimal value function may be obtained directly by solving the appropriate algebraic Riccati equation. We used the quadratic basis $\{p^2,v^2,pv\}$ on $\Omega=[-1,1]^2$ with initial policy $u_0=-p-v$. Since this basis spans the exact unconstrained LQR value function, SGA recovers $\hat V=V^*$ up to numerical tolerance, and the resulting value function satisfies the unconstrained HJB equation globally. The policy iteration terminates after $4$ iterations in approximately \qty{0.4}{\second}.

We used $x_0=[10,\,0]^\top$, $Q=I_2$, $R=1$, $\alpha=10$, with constraints $|v|\leq \qty{1}{\meter\per\second}$, and $|u|\leq \qty{1}{\newton}$. Since $\dot v=u$, the velocity constraints have relative degree one. Using the CBFs $\bar B=v_{\max}-v$ and $\unbar B=v-v_{\min}$, we obtain the input bound $-\alpha(v-v_{\min})\leq u\leq\alpha(v_{\max}-v)$.

For comparison, the constrained LQR benchmark uses the infinite-horizon LQR value function $V_\infty(x)=x^\top P_\infty x$ as the terminal cost of a one-step QP with $\Delta t=\qty{0.1}{\second}$, enforcing input and next-step velocity bounds. The resulting scalar-input parametric QP is solved offline by active-set enumeration and KKT conditions, yielding an exact piecewise-affine feedback law evaluated online by region lookup. The MPC uses the same stage cost and constraints with horizon $N=20$, $\Delta t=\qty{0.1}{\second}$, and terminal cost $x_N^\top P_\infty x_N$. No terminal constraint is imposed; a condensed QP is solved at each sample and only the first input is applied.

\begin{figure}[htp]
	\centering    \includegraphics{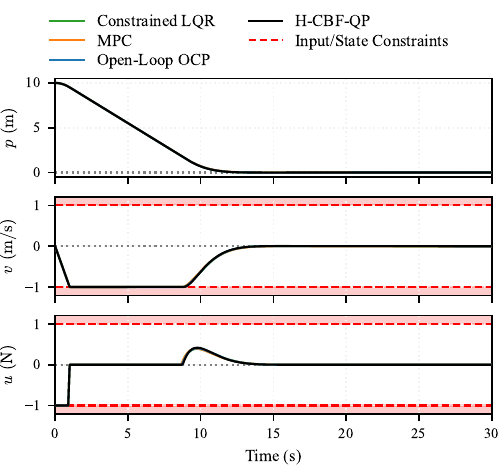}
	\caption{Position, velocity, and input trajectories for the 1D hovercraft. The trajectories of all controllers are superimposed.}
	\label{fig:hovercraft_Sim}
\end{figure}

As shown in \cref{fig:hovercraft_Sim}, the constrained LQR, MPC, open-loop OCP, and H-CBF-QP trajectories all coincide, with costs agreeing to within $10^{-3}$ ($J\approx398.34$). Thus, for this example, the H-CBF-QP reproduces the numerical constrained-optimal benchmark. The trajectory consists of an input-saturated phase, a velocity-saturated coasting phase with $\dot v=u=0$, and a terminal unconstrained LQR phase. The H-CBF-QP reproduces this structure through the input bound and relative-degree-one CBF constraints, while in the final unconstrained phase it recovers the LQR feedback by \cref{prop:recovery}.

%%%%%%%%%%%%%%%%%%%%%%%%%%%%%%%%%%%%%%%%%%%%%%%%%%%%%%%%%%%%%%%%%%%%%%%%%%%%%%%%
\subsection{Example 2 (Nonlinear): Spacecraft Attitude Control}
We apply the framework to spacecraft attitude control with reaction wheels, previously studied in~\cite{shahraki2025spacecraft} using CLF-CBF-QP. The online constrained model has state
$x=[\sigma^\top,\omega^\top,h_w^\top]^\top\in\R^9$, where
$\sigma\in\R^3$ denotes the Rodrigues Parameters (RPs),
$\omega\in\R^3$ is the body-frame angular velocity, and
$h_w\in\R^3$ is the reaction-wheel angular momentum. The input
$u\in\R^3$ is the control torque. Three identical, axially symmetric
reaction wheels are aligned with the spacecraft body-frame axes. The dynamics take the input-affine form \eqref{eq:dynamics}, with~\cite{shahraki2025spacecraft,schaub1996stereographic}
\begin{equation*}
f(x)\defeq
\begin{bmatrix}
G(\sigma)\omega\\
-J^{-1}[\omega]_\times(J\omega+h_w)\\
0_{3\times1}
\end{bmatrix},
\qquad
g(x)\defeq
\begin{bmatrix}
0_{3\times3}\\
J^{-1}\\
-I_3
\end{bmatrix},
\end{equation*}
where
$G(\sigma)=\tfrac12(I_3+[\sigma]_\times+\sigma\sigma^\top)$,
$J\in\R^{3\times3}$ is the spacecraft inertia matrix, and
$[\cdot]_\times$ denotes the skew-symmetric cross-product matrix.

For this example, we used SOS policy iteration, as described in \cref{sec:sos}.
For the offline value-function approximation, we use the reduced state $\xi=[\sigma^\top,\omega^\top]^\top\in\R^6$ and corresponding unconstrained attitude dynamics, introducing $h_w$ only in the online constrained model. The SOS approximation uses all even-degree monomials of degree $2$ and $4$ ($147$ monomials) in $\xi$ over $\Omega=[-1,1]^6$, with regional inequalities enforced via the $S$-procedure. Starting from $u_0=-\sigma-3\omega$, the MOSEK-based policy iteration terminates after $7$ iterations in approximately $\qty{7.6}{\second}$.

We consider a \emph{rest-to-rest maneuver} to the origin from
$x_0=[\sigma_0^\top,0^\top,0^\top]^\top$, where $\sigma_0=[0.312,\,-0.666,\,0.606]^\top$ corresponds to initial Euler angles\footnote{Using the ZYX extrinsic Euler angles convention.} $[80^\circ,\,-30^\circ,\,60^\circ]^\top$. The remaining parameters are $Q=I_6$, $R=I_3$,
\begin{equation*}
J=
\begin{bmatrix}
1.8140 & -0.1185 & 0.0275\\
-0.1185 & 1.7350 & 0.0169\\
0.0275 & 0.0169 & 3.4320
\end{bmatrix}
\,\unit{kg.m^2},
\end{equation*}
and $|u_i|\leq \qty{0.123}{\newton\meter}$.

\paragraph*{Case~1 (CBF): Reaction wheel momentum constraints}
The reaction-wheel momentum satisfies $\dot h_w=-u$ and is constrained by $|h_{w,i}|\leq\qty{0.4}{\newton\meter\second}$. Each momentum constraint therefore has relative degree one. Although $h_w$ is excluded from the offline SOS approximation, it is included in the online constrained dynamics and CBF-QP. Including all nine states on the SOS basis would significantly increase the number of monomials, making the offline formulation substantially less tractable. Defining the componentwise CBFs $\bar B_i=h_{w,\max,i}-h_{w,i}$ and $\unbar B_i=h_{w,i}-h_{w,\min,i}$ yields
$
-\alpha(h_{w,\max}-h_w)
\leq u \leq
\alpha(h_w-h_{w,\min})
$~\cite{shahraki2025spacecraft},
with $\alpha=10$.

\begin{figure}[htp]
	\centering\includegraphics{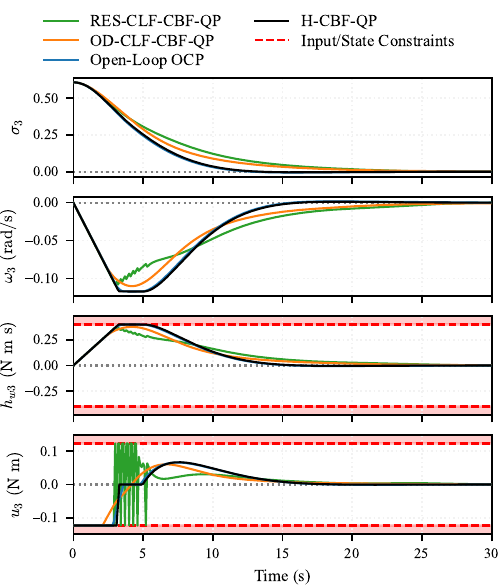}
	\caption{Representative attitude RPs, angular velocity, reaction-wheel momentum, and control torque responses for the spacecraft under momentum constraints. Only the 3\textsuperscript{rd} components of $\sigma$, $\omega$, $h_w$, and $u$ are shown.}
	\label{fig:spacecraft_sim_hw}
\end{figure}

\cref{fig:spacecraft_sim_hw,table:spacecraft_performance} show the results with active reaction-wheel momentum constraints. The H-CBF-QP closely tracks the open-loop OCP across the state and input trajectories, while both CLF-CBF-QP controllers exhibit larger transients and slower convergence. In this example, the relative-degree-one momentum constraint reduces to an instantaneous affine bound on $u$, and the H-CBF-QP closely reproduces the constraint-active behavior of the OCP. Once the momentum constraint is no longer active, \cref{prop:recovery} gives $\bar u=\hat u$. This close trajectory agreement is reflected in the near-OCP cost reported in \cref{table:spacecraft_performance}.

\paragraph*{Case~2 (HOCBF): Forbidden pointing constraint}
Consider an instrument with body-frame boresight direction $b\in\R^3$ that must avoid a bright celestial object with inertial direction $n\in\R^3$ by at least a half-cone angle $\theta$. The pointing constraint is
$
B(\sigma)\defeq\cos\theta-b^\top R(\sigma)n\geq0
$~\cite{wang2025control}, where the rotation matrix from the inertial to body frame is
$
R(\sigma)
=
\frac{1}{1+\sigma^\top\sigma}
\left[
(1-\sigma^\top\sigma)I_3
+2\sigma\sigma^\top
-2[\sigma]_\times
\right]
$~\cite{schaub1996stereographic}.

Since $B$ depends only on $\sigma$ and $\dot\sigma=G(\sigma)\omega$,
$L_gB=0$. Moreover, $L_gL_fB\neq0$, so $B$ has relative degree two. The HOCBF construction of \cref{def:hocbf} gives
\begin{align*}
\psi_0 &= B, \qquad
\psi_1 = L_fB+\alpha_1B,\\
\psi_2 &= L_f^2B+L_gL_fB\,u
+(\alpha_1+\alpha_2)L_fB+\alpha_1\alpha_2B,
\end{align*}
where
$L_fB=\frac{\partial B}{\partial\sigma}G(\sigma)\omega$ and
$L_gL_fB=\frac{\partial B}{\partial\sigma}G(\sigma)J^{-1}$.
The constraint $\psi_2\geq0$ is affine in $u$ and is imposed directly in \eqref{eq:H_cbf_qp}. We use
$b=[0,\,0,\,1]^\top$,
$n=[-0.47,\,-0.19,\,0.86]^\top$ (normalized),
$\theta=\ang{15}$, and $\alpha_1=\alpha_2=1$.

\begin{figure}[htp]
	\centering
	\includegraphics{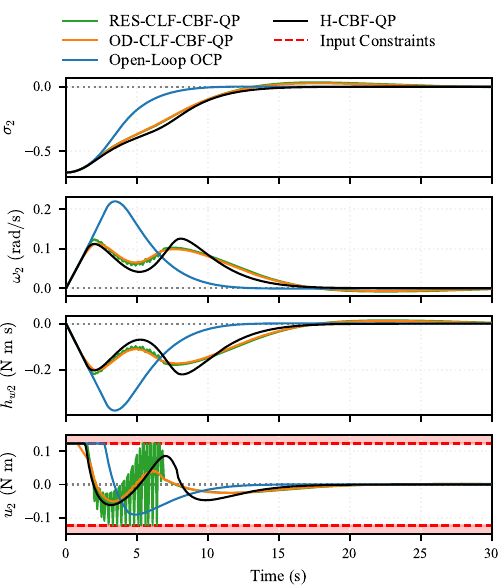}
	\caption{Representative attitude RPs, angular velocity, reaction-wheel momentum, and control torque responses under the pointing constraint. Only the 2\textsuperscript{nd} components of $\sigma$, $\omega$, $h_w$, and $u$ are shown.}
	\label{fig:spacecraft_sim_pointing}
\end{figure}

\vspace{-2mm}\begin{figure}[htp]
	\centering
	\includegraphics{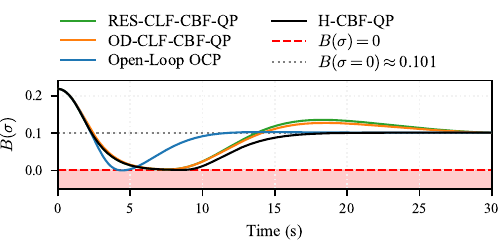}
	\caption{Pointing-constraint function $B(\sigma)$ for the constrained controllers. All trajectories satisfy $B(\sigma)\geq0$.}
	\label{fig:pointing_constraint}
\end{figure}

\begin{figure}[htp]
	\vspace{1mm}
	\centering
	\includegraphics{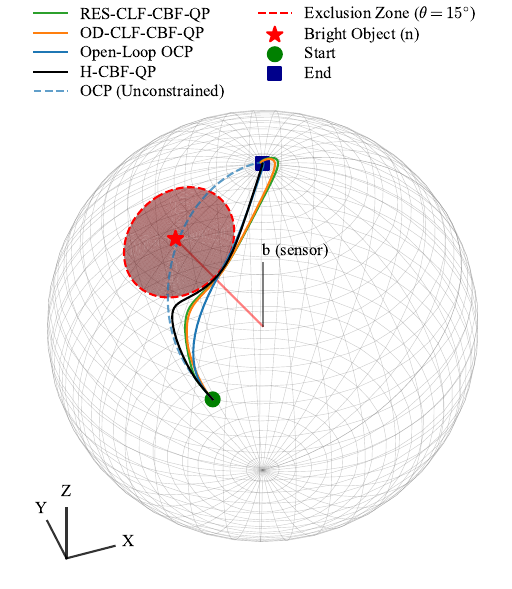}\vspace{-2mm}
	\caption{Sensor boresight trajectories on the celestial sphere. The red region denotes the exclusion zone of half-angle $\theta=\ang{15}$ around the bright object. All constrained trajectories remain outside the exclusion zone.}
	\label{fig:pointing_sphere_3d}
\end{figure}

\begin{table}[ht!]
\renewcommand{\arraystretch}{1.2}
\centering
\caption{Cost and computation time for the spacecraft attitude-control problem. The reported times are the total computation times for the online feedback controllers. Case~1: reaction-wheel momentum constraints. Case~2: forbidden pointing constraint.}
\label{table:spacecraft_performance}
{\setlength{\tabcolsep}{4pt}%
\begin{tabular}{
    l
    S[table-format=1.4]
    S[table-format=1.2]
    S[table-format=1.4]
    S[table-format=1.2]
}
\toprule
    & \multicolumn{2}{c}{Case~1} & \multicolumn{2}{c}{Case~2} \\
    \cmidrule(lr){2-3} \cmidrule(lr){4-5}
    Control Method
    & \multicolumn{1}{c}{Cost $J$}
    & \multicolumn{1}{c}{Time (s)}
    & \multicolumn{1}{c}{Cost $J$}
    & \multicolumn{1}{c}{Time (s)} \\
\midrule
    RES-CLF-CBF-QP & 3.5841 & 0.23 & 4.1937 & 0.23 \\
    OD-CLF-CBF-QP  & 3.4755 & 0.24 & 4.0626 & 0.24 \\
    \textbf{H-CBF-QP} & \textbf{3.2249} & \textbf{0.22} & \textbf{3.9842} & \textbf{0.22} \\
\midrule
    Open-Loop OCP
        & 3.2211
        & \multicolumn{1}{c}{--}
        & 3.3768
        & \multicolumn{1}{c}{--} \\
\bottomrule
\end{tabular}}
\end{table}

\cref{fig:spacecraft_sim_pointing,fig:pointing_constraint,fig:pointing_sphere_3d,table:spacecraft_performance}
summarize the results with the forbidden pointing constraint. The H-CBF-QP remains qualitatively closer to the OCP trajectory than the CLF-CBF-QP benchmarks, which exhibit larger transients; the RES-CLF-CBF-QP also exhibits noticeable input chattering associated with its prescribed CLF decay. \Cref{fig:pointing_constraint,fig:pointing_sphere_3d} confirm satisfaction of the pointing constraint throughout the maneuver.

Unlike Case~1, the constrained OCP begins deviating from the unconstrained trajectory before the HOCBF constraint becomes active, thereby anticipating the future exclusion-zone encounter. The H-CBF-QP, by contrast, has a zero-horizon online structure and uses an unconstrained value function: it follows the unconstrained value-function policy whenever feasible and modifies it only through the instantaneous HOCBF constraint. Consequently, future constraint encounters are not explicitly anticipated, which is consistent with the observed cost gap relative to the OCP in \cref{table:spacecraft_performance}. The HOCBF gains were not optimized; since performance depends on $\alpha_1$ and $\alpha_2$, further tuning may reduce this gap.

The source code reproducing all figures is available at \url{https://github.com/QCGroup/hcbf}.

%%%%%%%%%%%%%%%%%%%%%%%%%%%%%%%%%%%%%%%%%%%%%%%%%%%%%%%%%%%%%%%%%%%%%%%%%%%%%%%%
\section{Conclusion and Future Work}\label{sec:conclusion}
We presented a two-stage framework for constrained feedback control of input-affine nonlinear systems in which an approximate value function is computed offline via HJB-based policy iteration and constraints are enforced online through the H-CBF-QP. The method decouples performance optimization from safety enforcement, allowing constraints to be modified without recomputing the value function. Under suitable conditions, we established safety, approximate optimality, and local asymptotic stability. The examples show near-OCP performance when the constrained-optimal trajectory aligns with the unconstrained value-function policy, while future constraint encounters can introduce a performance gap due to the zero-horizon CBF-QP structure and the use of the unconstrained offline value function.

Future work includes finite-horizon extensions, characterizing when constrained-optimal behavior is recovered, and incorporating limited lookahead to anticipate future constraints.

\bibliographystyle{IEEEtran}
\bibliography{bibabbrv}
\end{document}